\documentclass[draft,runningheads,a4paper,envcountsame,envcountsect]{llncs}

\usepackage{amsmath,amsthm,svninfo}
\usepackage{mathtools,amssymb,cutwin}
\usepackage{mathrsfs,dsfont}
\usepackage{pgf,tikz,braids}
\usetikzlibrary{arrows,automata,shapes,backgrounds,shapes.symbols}
\usetikzlibrary{positioning,fadings}
\usetikzlibrary{decorations.pathmorphing}
\usetikzlibrary{snakes}
\usetikzlibrary{trees,decorations}
\usetikzlibrary{calc}
\usetikzlibrary{shapes.callouts,shapes.arrows}
\usepackage{verbatim}
\usepackage[]{algorithm2e}
\usepackage{wrapfig}
\usepackage{xspace}

\usepackage[format=plain,labelfont=bf,up,textfont=it,up]{caption}
\usepackage{subfig}

\newcommand\FR{\textsf{FR}\xspace}
\newcommand\GAP{\textsf{GAP}\xspace}

\theoremstyle{plain}
\newtheorem*{theorem*}{Theorem}

\newcommand{\lacroix}{\tikz[baseline=-.5ex]{\draw[->,>=latex] (0,0) -- (4ex,0); \draw[->,>=latex] (1.8ex,2ex) -- (1.8ex,-2ex);}}

\newcommand{\auta}{{\mathcal M}}
\newcommand{\outs}[1]{{#1}^{\mathrm{out}}}

\newcommand{\lab}[3]{#1:{#2},{#3}}
\newcommand{\labb}[3]{#1\!:\!{\bf #2},{\bf #3}}

\newcommand{\presp}[1]{\langle{#1}\rangle_+}

\newcommand{\detee}[1]{\mathrm{det}\left({#1}\right)}

\newcommand\N{{\mathbb N}}
\newcommand\Z{{\mathbb Z}}
\newcommand\R{{\mathbb R}}
\newcommand\SPol{{\mathrm{SPol}}}
\newcommand\Pol{{\mathrm{Pol}}}

\newcommand\FEnd{{\mathrm{FEnd}}}
\newcommand\FAut{{\mathrm{FAut}}}
\newcommand\sss{r}
\newcommand\ttt{s}

\newcommand\WP{{\textsc{\texttt{Word Problem}}}\xspace}
\newcommand\OP{{\textsc{\texttt{Order Problem}}}\xspace}

\newcommand\lcm{{\rm lcm}}
\newcommand\resp{\hbox{\emph{resp.}}\xspace}
\newcommand\ie{\hbox{\emph{i.e.}}\xspace}

\newcommand\unit{\scalebox{.9}{$\mathds{1}$}}

\newcommand\dindex{m}
\newcommand\dperiod{\ell}

\newcommand{\Tree}{\Sigma^*}
\newcommand{\FE}{{\rm FEnd}(\Tree)}
\newcommand{\SP}[1]{{\mathrm{SPol}}\!\left({#1}\right)}
\newcommand{\SE}[1]{{\mathrm{SExp}}\!\left({#1}\right)}

\newcommand{\pnorm}[1]{\| #1\|_{\mathrm p}}
\newcommand{\snorm}[1]{\| #1\|_{\mathrm e}}
\newcommand{\size}[1]{\| #1\|}
\newcommand{\card}[1]{\# #1}

\newcommand{\sect}[2]{{{#1}@{#2}}}
\newcommand{\act}[2]{{#2}^{#1}}

\newcommand{\mot}[1]{\mathbf{#1}}

\newcommand{\NFA}{\textbf{NFA}\xspace}
\newcommand{\DFA}{\textbf{DFA}\xspace}
\makeatletter 

\tikzset{curlybrace/.style={rounded corners=2pt,line cap=round}}%  
\pgfkeys{%
/curlybrace/.cd,%
tip angle/.code     =  \def\cb@angle{#1},
/curlybrace/.unknown/.code ={\let\searchname=\pgfkeyscurrentname
                              \pgfkeysalso{\searchname/.try=#1,
                              /tikz/\searchname/.retry=#1}}}  
\def\curlybrace{\pgfutil@ifnextchar[{\curly@brace}{\curly@brace[]}}%

\def\curly@brace[#1]#2#3#4{% 
\pgfkeys{/curlybrace/.cd,
tip angle = 1.75}% 
\pgfqkeys{/curlybrace}{#1}% 
\ifnum 1>#4 \def\cbrd{0.105} \else \def\cbrd{0.075} \fi
\draw[/curlybrace/.cd,curlybrace,#1]  (#2:#4-\cbrd) -- (#2:#4) arc (#2:{(#2+#3)/2-\cb@angle}:#4) --({(#2+#3)/2}:#4+\cbrd) coordinate (curlybracetipn);
\draw[/curlybrace/.cd,curlybrace,#1] ({(#2+#3)/2}:#4+\cbrd) -- ({(#2+#3)/2+\cb@angle}:#4) arc ({(#2+#3)/2+\cb@angle} :#3:#4) --(#3:#4-\cbrd);
}

\makeatother

\begin{document}

\title{A new hierarchy\\ for automaton semigroups}
\author{Laurent Bartholdi \inst{1},
Thibault Godin\inst{2},
Ines Klimann\inst{3}, and
Matthieu Picantin\inst{3}}

\institute{\'Ecole Normale Sup\'erieure, Paris, and
Georg-August-Universit\"at zu G\"ottingen,\\
\email{laurent.bartholdi@ens.fr},\\
\and University of Turku, Turku, \\ \email{thibault.godin@utu.fi}, supported by the Academy of Finland grant 296018
\and
IRIF, UMR 8243 CNRS \& Universit\'e Paris Diderot,\\ \email{klimann@irif.fr, picantin@irif.fr}\\
}

\maketitle

\begin{abstract}
  We define a new strict and computable hierarchy for the family of
  automaton semigroups, which reflects the various asymptotic
  behaviors of the state-activity growth. This hierarchy extends that
  given by Sidki for automaton groups, and also gives new insights into
  the latter.  Its exponential part coincides with a notion of entropy
  for some associated automata.

  We prove that the \OP is decidable when the state-activity is
  bounded. The \OP remains open for the next level of this hierarchy,
  that is, when the state-activity is linear. Gillibert showed that it
  is undecidable in the whole family.

  The former results are implemented and will be available in the
  \GAP\ package~\FR\ developed by the first author.
\end{abstract}

 
%-------------------------------------------------------------------------------------------------------------------------
\section{Introduction}\label{s:intro}

The family of automaton groups and semigroups has provided a wide playground
to various algorithmic problems in computational (semi)group theory~\cite{AKLMP12,Bar16,FR,BM17,BS10,BBSZ13,Gil14,Gil18,GK18,KMP12}.
While many undecidable questions in the world of (semi)groups remain
undecidable for this family, the underlying Mealy automata provide a combinatorial leverage
to solve the \WP for this family, and various other problems in some important subfamilies.
Recall that a Mealy automaton is a letter-to-letter, complete,
deterministic transducer with same input and output alphabet, so each of its states
induces a transformation from the set of words over its alphabet into itself.
Composing these Mealy transformations leads to so-called automaton
(semi)groups, and the \WP can be solved using a classical technique of minimization.

The \OP is one of the current challenging problems in computational
(semi)group theory. On the one hand, it was proven to be undecidable for automaton
semigroups by Gillibert~\cite{Gil14}.
On the other hand, Sidki introduced a polynomial hierarchy for
invertible Mealy transformations in~\cite{Sidki00} and, with
Bondarenko, Bondarenko, and Zapata in~\cite{BBSZ13}, solved the \OP
for its lowest level (bounded invertible automata).

\medskip Our main contributions in this paper are the following: an
activity-based hierarchy for possibly non-invertible Mealy
transformations (Section~\ref{s:hierarchy}), extending Sidki's
construction~\cite{Sidki00} to non-necessarily-invertible transformations;
and a study of the algorithmic properties in the lowest level of the
hierarchy, namely transducers with bounded activity. We prove:
\begin{theorem*}[see Section~\ref{s:order}]
  The \OP is decidable for bounded Mealy transformations; namely,
  there is an algorithm that, given a bounded initial Mealy automaton,
  decides whether the transformation $\tau$ that it defines has
  infinite order, and if not finds the minimal $r>s$ such that
  $\tau^r=\tau^s$.
\end{theorem*}

Our strategy of proof follows closely that of Sidki's~\cite{Sidki00}
and Bondarenko, Bondarenko, Sidki and Zapata~\cite{BBSZ13}, with some
crucial differences. On the one hand, a naive count of the number of
non-trivial states of a transformation does not yield a useful
invariant, nor a hierarchy stable under multiplication; on the other
hand, the structure of cyclic semigroups
($\langle a \mid a^m =a^{m+n}\rangle$ has index $m$ and period
$n$) is more complex than that of
cyclic groups ($\langle a\mid a^m\rangle$ has order $m$).

%-------------------------------------------------------------------------------------------------------------------------
\section{Notions from automata and graph theory}\label{s:graph}

This section gathers some basics about automata, especially some links
between automata, Mealy automata, automaton semigroups, and
finite-state transformations.  We refer the reader to handbooks for
graph theory~\cite{Rig16}, automata theory~\cite{Sak09}, and automaton
(semi)groups~\cite{BS10}.

\medskip
A \emph{non-deterministic finite-state automaton} (\NFA for short) is
given by a directed graph with finite vertex set $Q$, a set of edges
$\Delta$ labeled by an alphabet $\Sigma$, and two distinguished subsets
of vertices \(I\subseteq Q\) and \(F\subseteq Q\). The vertices of the
graph are called \emph{states} of the automaton and its edges are
called \emph{transitions}. The elements of \(I\) and \(F\) are called
respectively \emph{initial} and \emph{final} states.  A transition
from the state~\(p\) to the state~\(q\) with label~\(x\) is denoted
by~\(p\xrightarrow{\ x\ }q\).

\smallskip
A \NFA is \emph{deterministic}---\DFA for short---(\resp
\emph{complete}) if for each state~\(q\) and each letter~\(x\), there
exists at most (\resp at least) one transition from \(q\) with
label~\(x\).  Given a word~\(\mot{w}=w_1w_2\cdots w_n \in \Sigma^*\)
(where the \(w_i\) are letters), a \emph{run with label~\(\mot{w}\)}
in an automaton (\NFA\ or \DFA) is a sequence of consecutive
transitions
\[q_1\xrightarrow{\ w_1\ }q_2\xrightarrow{\ w_2\ }
  q_3\rightarrow\cdots \rightarrow q_{n}\xrightarrow{\ w_n\
  }q_{n+1}\:.
\]
Such a run is \emph{successful} whenever \(q_1\) is an initial state
and \(q_{n+1}\) a final state.  A word in $\Sigma^*$ is
\emph{recognized} by an automaton if it is the label of at least one
successful run. The \emph{language} recognized by an automaton is the
set of words it recognizes.  A \DFA is \emph{coaccessible} if each
state belongs to some run ending at a final state.

Let~\(\mathcal A\) be a \NFA with stateset~\(Q\). The Rabin--Scott
powerset construction~\cite{RaSc59} returns, in a nutshell, the
(co)accessible \DFA---denoted by~\(\detee{\mathcal A}\)---with states
corresponding to subsets of~\(Q\), whose initial state is the subset
of all initial states of \(\mathcal A\) and whose  final states are
the subsets containing at least on final state of~\(\mathcal A\); its transition labeled by~\(x\) from a state~\(S\subseteq 2^Q\) leads to the state~\(\{q\:|\: \exists p \in S, p\xrightarrow{x}q \text{ in }\mathcal A\}\).
Notice that the size of the resulting~\DFA might therefore be exponential in the size of the original~\NFA.

\medskip
The \emph{language} of a \NFA is the subset of $\Sigma^*$ consisting
in the words recognized by it. Given a language~\(L\subseteq\Sigma^*\), its
\emph{entropy} is
\[h(L) = \lim_{\ell \to \infty} \quad \frac{1}{\ell}\ {{\log \#\left(L
        \cap \Sigma^\ell \right)}} .
\]
This quantity appears in various situations, in particular for
subshifts~\cite{LiMa95} and finite-state automata~\cite{ChMi58}. We
shall see how to compute it with matrices.

To a \NFA $\mathcal A$, associate its \emph{transition matrix}
\(A = \{A_{i,j}\}_{i,j} \in \N^{n\times n}\) where~\(A_{i,j}\) is the
number of transitions from~\(i\) to~\(j\). Let furthermore $v\in\N^n$
be the row vector with `$1$' at all positions in $I$ and $w\in\N^n$ be
the column vector with `$1$' at all positions in $F$. Then
$v A^\ell w$ is the number of successful runs in $\mathcal A$ of
length $\ell$. Assuming furthermore that $\mathcal A$ is
deterministic, $v A^\ell w$ is the cardinality of $L\cap\Sigma^\ell$.

Since the transition matrix of an automaton $\mathcal A$ is
non-negative, it admits a positive real eigenvalue of maximal absolute
value, which is called its Perron-Frobenius eigenvalue and is
written~\(\lambda(\mathcal A)\). Assuming therefore that $\mathcal A$ is coaccessible, we get
\begin{proposition}{\rm\cite[Theorem~1.2]{Shu08a}}\label{prop-entraut}
  Let \(\mathcal{A}\) be a coaccessible \DFA recognizing the language
  \(L\). Then we have \(h(L) = \log \lambda({\mathcal{A}})\).
\end{proposition}

\subsection{Mealy transducers}
A \emph{Mealy automaton} is a \DFA over an alphabet of the form
$\Sigma\times\Sigma$. If an edge's label is $(x,y)$, one calls $x$ the
\emph{input} and $y$ the \emph{output}, and denotes the transition
by~\(p\xrightarrow{\ x|y\ }q\). Such a Mealy automaton~\(\auta\) is
assumed to be complete and deterministic in its inputs: for every
state~\(p\) and letter~\(x\), there exists exactly one transition
from~\(p\) with input letter~\(x\). We denote by~\(\act{p}{x}\) its
corresponding output letter and by~\(\sect{p}{x}\) its target state,
so we have
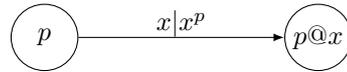
\begin{center}
\begin{tikzpicture}[->,>=latex,node distance=36mm,inner sep=0pt]
  \node[state] (0) {\(p\)};
  \node[state,right of=0] (1) {\(\sect{p}{x}\)};
  \path (0) edge node[above] {\(x|\act{p}{x}\)} (1);
\end{tikzpicture}
\end{center}

In this way, states act on letters and letters on states.
Such actions can be composed in the following way: let
\(\mot{q}\in Q^*\), \(p\in 
Q\), \(\mot{u}\in\Sigma^*\), and \(x\in\Sigma\), we have
\[\act{\mot{q}p}{x} = \act{p}{(\act{\mot{q}}{x})}\quad\text{ and }\quad
\sect{p}{(\mot{u}x)} = \sect{(\sect{p}{\mot{u}})}{x}\:.\]

We extend recursively the actions of states on letters and of letters on states (see below left).
Compositions can be more easily understood via an alternative representation by a~\emph{cross-diagram} \cite{AKLMP12} (below right).

For all~\(x \in \Sigma\), \(\mot{u} \in \Sigma^*\), \(p\in Q\), \(\mot{q}\in Q^*\), we have:
\begin{center}
\begin{minipage}{.4\textwidth}
\[\act{\mot{q}}{(\mot{u}x)} =
\act{\mot{q}}{\mot{u}}\act{\sect{\mot{q}}{\mot{u}}}{x}\]
and
\[\sect{(\mot{q}p)}{\mot{u}} =
\sect{\mot{q}}{\mot{u}}\cdot\sect{p}{\act{\mot{q}}{\mot{u}}}\:.\]
\end{minipage}
\begin{minipage}{.4\textwidth}
\vspace*{-5pt}
\[\begin{array}{ccccc}
 & \mot{q} & & p   &\\
\mot{u} & \lacroix & \act{\mot{q}}{\mot{u}} & \lacroix & \act{\mot{q}p}{\mot{u}} \\
  &\sect{\mot{q}}{\mot{u}}  & & \sect{p}{\act{\mot{q}}{\mot{u}}} & \\
x      &\lacroix  & \act{\sect{\mot{q}}{\mot{u}}}{x} & & \\
  &\sect{\mot{q}}{\mot{u}x} & & &
\end{array}\]
\end{minipage}
\end{center}

The mappings defined above are length-preserving and prefix-preserving. Note that in particular the image of the empty word is itself.

\smallbreak From an algebraic point of view, the composition gives a semigroup structure to the set
of transformations~\(\mot{u}\mapsto\act{\mot{q}}{\mot{u}}\) for~\(\mot{q}\in Q^*\). This
semigroup is called \emph{the semigroup generated by~\(\auta\)} and
denoted by~\(\presp{\auta}\).
An \emph{automaton semigroup} is a semigroup which can be generated by
a Mealy automaton. Any element of such an automaton semigroup
induces a so-called \emph{finite-state} transformation.

\smallbreak Conversely, for any transformation~\(t\) of~\(\Sigma^*\)
and any word~\(\mot{u} \in \Sigma^*\), we denote by
\(\act{t}{\mot{u}} \) the image of~\(\mot{u}\) by~\(t\), and by
\(\sect{t}{\mot{u}}\) the unique transformation~\(s\) of~\(\Sigma^*\)
satisfying \(\act{t}{(\mot{uv})}=\act{t}{\mot{u}}\act{s}{\mot{v}}\)
for
any~\(\mot{v}\in\Sigma^*\).
Whenever~\(Q(t)=\{\sect{t}{\mot{u}}:\mot{u}\in\Sigma^*\}\) is finite,
\(t\) is said to be~\emph{finite-state} and admits a unique (minimal)
\emph{associated Mealy automaton}~\(\auta_t\) with stateset~\(Q(t)\).

We also use the following convenient notation to define a finite-state
transformation $t$: for each $u\in Q(t)$, we write an equation (traditionally called \emph{wreath recursion} in the algebraic theory of automata) of the
following form
\[u=(\sect{u}{x_1},\ldots,\sect{u}{x_{|\Sigma|}})\sigma_u,\]
where~\(\sigma_u=[\act{u}{x_1},\ldots,\act{u}{x_{|\Sigma|}}]\) denotes
the transformation on~\(\Sigma\) induced by~\(u\).

\smallbreak We consider the semigroup~\(\FE\) of those finite-state transformations of~\(\Sigma^*\).

\begin{example}\label{e:spol-zero} The transformation~\(t_0=(\unit,t_0)[2,2]\) belongs to~\(\FEnd(\{1,2\}^*)\) with~\(Q(t_0)=\{\unit,t_0\}\).
See Examples~\ref{e:spol} and~\ref{e:semST} for further details about~\(t_0\).
\end{example}

\begin{example}\label{e:activity-zero} The transformation~\(p=(q,r)\) with~\(q=(r,\unit)\)
and~\(r=(r,r)[2,2]\) also belongs to~\(\FEnd(\{1,2\}^*)\) with~\(Q(p)=\{\unit,p,q,r\}\).
See Figure~{\bf\ref{f:activity}} for~\(\auta_p\).
%MealyMachine([[1,1],[3,4],[4,1],[4,4]],[[1,2],[1,2],[[1,2],[2,2]]);
\end{example}

%-------------------------------------------------------------------------------------------------------------------------
\section{An activity-based hierarchy for~$\FE$}\label{s:hierarchy}

In this section we define a suitable notion of activity for finite-state transformations, together with two norms, from which we build a new hierarchy.
We will prove its strictness and its computability in Section~\ref{s:charac}.

\medskip  
For any element~$t\in\FE$, we define its \emph{activity} (see Figure~{\bf \ref{fig-acttree}}) as
\[\alpha_t:n\longmapsto\#\{\mot{v}\in\Sigma^n:\exists \mot{u}\in\Sigma^n,~\sect{t}{\mot{u}}\not=\unit~\hbox{and}~\act{t}{\mot{u}}=\mot{v}\}. \label{eq:act}\]

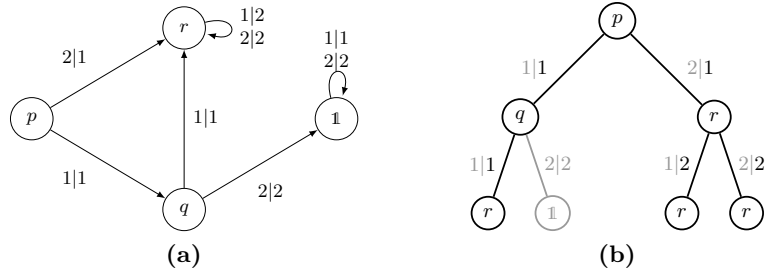
\begin{figure}
	\begin{center}
			\subfloat[]{\label{f:activity}
			\scalebox{0.8}{
			\begin{tikzpicture}[->,>=latex,node distance=23mm,inner sep=0,minimum size=7mm]
				\begin{scope}
				\node[draw,circle,minimum width=19pt] (p) {$p$};
				\node[draw,circle,minimum width=19pt,above  right  = 10mm and 20mm of p] (r) {$r$};
				\node[draw,circle,minimum width=19pt,below right = 10mm and 20mm of p] (q) {$q$};
				\node[draw,circle,minimum width=19pt,below right = 10mm and 20mm of r] (one) {$\unit$};

				\path	(r)	edge[loop right] 		node[right,align=center]{\(1|2\)\\ \(2|2\)}	(r)
					(p)	edge					node[below left,pos=.4]{\(1|1\)}	(q)
					(p)	edge[] 		node[above left, pos=.4]{\(2|1\)}	(r)
					(q)	edge				node[right]{\(1|1\)}	(r)
					(q)	edge				node[below right,pos=.4]{\(2|2\)}	(one)
					(one)	edge[loop above]		node[above,align=center]{\(1|1\)\\\(2|2\)}	(one)
			 ;
				\end{scope}
			\end{tikzpicture}
			}
			}
	\quad\quad\quad\subfloat[]{\label{fig-treet0}
	\scalebox{.85}{
	\begin{tikzpicture}[thick,level distance=1.5cm,
level 1/.style={sibling distance=3cm},
level 2/.style={sibling distance=1cm},
tree node/.style={circle,draw},
every child node/.style={tree node}]
		\begin{scope}
		\node[tree node] {\color{black}  \(p\)} [sibling distance=1.2cm]
		  child { node[] {\color{black} \(q\)} [sibling distance=.6cm]
				  child { node[] (c) { \color{black}  \(r\)} [sibling distance=.3cm]  edge from parent node[left,draw=none] {\(\color{black!40}{1}|\color{black}{1}\)} 
				        }
				  child { node[black!40,inner sep=2.75pt] {\(\unit\)} [sibling distance=.3cm,black!40] edge from parent node[right,draw=none] {\(\color{black!40}{2}|2\)} 
				        }
				edge from parent node[left=0.2cm,draw=none] {\(\color{black!40}{1}|\color{black}{1}\)} 
				}
		  child { node[] {\color{black}  \(r\)} [sibling distance=.6cm]
				  child { node[] {\color{black}  \(r\)} [sibling distance=.3cm] edge from parent node[left,draw=none] {\(\color{black!40}{1}|\color{black}{2}\)} 
				        }
				  child { node[] (a) {\color{black} \(r\)} [sibling distance=.3cm]	 edge from parent node[right,draw=none] {\(\color{black!40}{2}|\color{black}{2}\)} 			          
				        }
				edge from parent node[right=0.2cm,draw=none] {\(\color{black!40}{2}|\color{black}{1}\)} 
				};
		\end{scope}
		\end{tikzpicture}
		}
}%
\caption{ {\bf\small(a)} An example of a transformation~\(p\) with~\(\alpha_p(1)=1\) and~\(\alpha_p(2)=2\). {\bf\small(b)}  The transformation \(p\) induces~\(3\) nontrivial  transformations on level~\(2\): the leftmost one is associated with the output~\(11\), the middle right one with~\(12\) and the rightmost one with~\(12\), hence nontrivial transformations can  be reached by runs with only two different output words. }\label{fig-acttree}
\end{center}
\end{figure}

We next define two norms on~$\FE$.
When~$\alpha_t$ has polynomial growth, namely when the set~$D=\{d:\lim_{n\rightarrow \infty}\frac{\alpha_t(n)}{n^d}=0\}$ is nonempty,
then we define $\pnorm{t}=\min D-1$. 
Otherwise, the value of~$\lim_{n\rightarrow\infty}\frac{\log\alpha_t(n)}{n}$ is denoted by~$\snorm{t}$.

\medskip We then define the following classes of finite-state transformations:
\begin{align*}
\SP{d}&=\{\ t\in\FE:\pnorm{t}\leq d\ \}\\ \hbox{and}~~
\SE{\lambda}&=\{\ t\in\FE:\snorm{t}\leq\lambda\ \}\:.
\end{align*}

We shall see in Theorem~\ref{t:hierarchy} that these yield a strict and computable hierarchy for~$\FE$.
The following basic lemma is crucial:

\begin{lemma}\label{lem-subadditive} For each \(n\geq 0\), the map~$t\mapsto\alpha_t(n)$ is subadditive.
\end{lemma}

\begin{proof}
Assume~$s,t\in\FE$.

\begin{minipage}[t]{0.46\linewidth}
For any~$\mot{u}\in\Sigma^n$ with~$\sect{(st)}{\mot{u}}\not=\unit$,
we~have either~$\sect{s}{\mot{u}}\not=\unit$ or~$\sect{t}{\act{s}{\mot{u}}}\not=\unit$.\\
We deduce $\alpha_{st}(n)\leq\alpha_{s}(n)+\alpha_{t}(n)$ for each~$n\geq 0$.
\end{minipage}
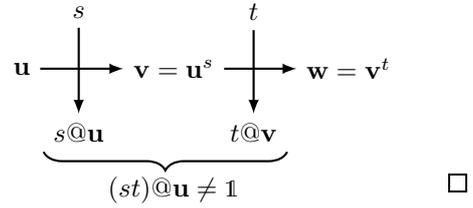
\hfill
\begin{minipage}[t]{0.5\linewidth}   
\vspace{-0.75cm}
\begin{tikzpicture}[thick,node distance=13mm]
	\begin{scope}[thick,node distance=16mm,xshift=50mm]
	\node (01) at (0,1.25) {$\mot{u}$};
	\node (10) at (.75,2)  {$s$};
	\node (12) [below of=10] {$\sect{s}{\mot{u}}$};
	\node (11) [node distance=20mm,right of=01] {$\mot{v}=\act{s}{\mot{u}}$};
	\node (20) [node distance=23mm,right of=10] {$t$};
	\node (22) [node distance=16mm,below of=20] {$\sect{t}{\mot{v}}$};
	\node (31) [node distance=23mm,right of=11] {$\mot{w}=\act{t}{\mot{v}}$};
	\node (44) [node distance=16mm,below of=11] {$\sect{(st)}{\mot{u}}\not=\unit$};
	\path[->,>=latex]
		(01)	edge  (11)
		(10)	edge  (12)
		(11)	edge  (31)
		(20)	edge  (22);
	\draw[decorate,decoration={brace,amplitude=7pt,aspect=.5}]   (22.south east)--(12.south west);
	\end{scope}
\end{tikzpicture}

\vspace*{-15pt}
\end{minipage}
\end{proof}

We deduce that \(\pnorm{.}\) and~\(\snorm{.}\) are respectively \(+\)- and~\(\max\)-subadditive.

\begin{proposition}\label{p:closed} Let~\(\Sigma\) be an alphabet.
For every integer~\(d\geq -1\), \(\SP{d}\) is a subsemigroup of~\(\FE\).
So is~\(\SE{\lambda}\) for every~\(0 \leq \lambda\leq\#\Sigma\).
\end{proposition}

As an easy corollary of Proposition~\ref{p:closed},
the subadditivity property allows us to compute the hierarchy class
of any given Mealy automaton by considering only its generators.

\begin{theorem}\label{t:hierarchy} Let~\(\Sigma\) be an alphabet.
The elements of the semigroup~\(\FE\) can be graded according to the following strict 
hierarchy:
for any~\(d_1,d_2\in \Z\) with~\(-1<d_1<d_2\) and
any~\(\lambda_1,\lambda_2\in \R\)
with~\(0<\lambda_1<\lambda_2<\#\Sigma\),  we have:
\[\SP{-1}\subsetneq\cdots\subsetneq\SP{d_1}\subsetneq\cdots\subsetneq\SP{d_2}\subsetneq\cdots\subsetneq\SE{0}\]
\[\subsetneq\SE{\lambda_1}\subsetneq\cdots\subsetneq\SE{\lambda_2}\subsetneq\cdots\subsetneq\SE{\#\Sigma}.\]
\end{theorem}

The proof of the previous result is postponed to the end of Section~\ref{s:charac} on page~\pageref{proof:hierarchy}.

\medskip Sidki defined in~\cite{Sidki00} the activity of a finite-state automorphism~\(t\in\FAut(\Sigma^*)\) as~\[\theta_t:n\longmapsto\card{\{\mot{u}\in\Sigma^n:\:\sect{t}{\mot{u}}\not=\unit\}},\] and the corresponding classes~\(\Pol(d)\). Using this notion of activity~\(\theta\) for transformations leads inevitably to an impasse: the associated classes with fixed degree polynomial \(\theta\)-activity would be not closed under composition. However it is straightforward that our new notion
of activity~\(\alpha\) coincides with Sidki's activity~\(\theta\) in the case of automorphisms.

The class~\(\SE{0}\) coincides with the infinite union~\(\bigcup_{d \geq -1} \SP{d}\), whose corresponding automorphisms class is denoted by~\(\Pol{(\infty)}\) in \cite{Sidki00}.

%-------------------------------------------------------------------------------------------------------------------------
\section{Structural characterization of the activity norm}\label{s:charac}

From \cite{BBSZ13}, we know that the finite-state automorphisms which have polynomial activity are exactly those  whose underlying automaton does not contain entangled cycles (except on the trivial state). Moreover, the degree of the polynomial is given by the longest chain of cycles in the automaton. The first claim remains true for transformations, but things are a bit more involved for the second one (see Example~\ref{e:spol}).

\medskip To any minimal Mealy automaton~\(\auta\) with stateset~\(Q\) and alphabet~\(\Sigma\), we associate
its \emph{pruned output} automaton~\(\outs{\auta}\) defined as the \NFA\
with stateset~\(Q\smallsetminus\{\unit\}\) (all states being final) and alphabet \(\Sigma\),
and whose transitions are given, for~\(p,q\in Q\smallsetminus\{\unit\}\), by
\[p \xrightarrow{~y~} q \ \in\ \outs{\auta} \quad\Longleftrightarrow\quad p \xrightarrow{~x|y~} q \ \in\ \auta.\]
According to context, we shall identify a transformation~\(t\in\FE\), the state of~\(\auta_t\), and the corresponding state of~\(\outs{\auta}_t\).

\begin{lemma}\label{lem:struct}
The activity of a transformation~\(t \in \FE\) is the number of paths starting from~\(t\)
in the (non-complete) deterministic automaton~\(\detee{\outs{\auta}_t}\) constructed via the Rabin--Scott construction.
\end{lemma}

\begin{proof}
Let \(t \in \FE\) with~\(\auta_t\) its associated automaton.  Let us count the words~\(\mot{v}\in\Sigma^n\) for which there is a word~\(\mot{u}\in\Sigma^n\) with~\(\sect{t}{\mot{u}}\not=\unit\) and~\(\act{t}{\mot{u}}=\mot{v}\). For~\(n=1\), \(\alpha_t(1)\) is exactly the number of different outputs from the state~\(t\) that do not lead to a trivial state of~\(\auta_t\). Now for \(\mot{v}\in \Sigma^n\), if \(\mathcal{E}\) denotes the set of those
states accessible from~\(t\) by reading~\(\mot{v}\) (this corresponds to the Rabin--Scott powerset construction) in~\(\outs{\auta}_t\),
the number of ways to extend~\(\mot{v}\) without getting into a trivial state in~\({\auta}_t\) corresponds to the number of outputs of the state~\(\mathcal{E}\) in~\(\detee{\outs{\auta}_t}\), whence the result.
\end{proof}

Whether the activity of a given~\(t\in\FE\) is polynomial or
exponential can be decided by only looking at the cycle structure
of~\(\outs{\auta}_t\). Any cycle considered throughout this paper is
simple: no repetitions of vertices or edges are allowed. Two cycles
are \emph{coreachable} if there exists a path from any of them to the
other one. A chain of cycles is a sequence of cycles such that each cycle is reachable from its predecessor.

\begin{proposition}
A transformation~\(t \in \FE\) has exponential activity if and only if it can reach two coreachable cycles with distinct labels in~\({\outs{\auta}_t}\).
\end{proposition}

\begin{proof}
$(\Leftarrow)$ Assume that~\(t\) can reach a state~\(s \in \outs{\auta}_t\) that lies on two cycles with distinct labels.
There exist a word~\(\mot{u}\in\Sigma^*\) satisfying~\(\sect{t}{\mot{u}}=s\)
and two words~\(\mot{v},\mot{w}\in\Sigma^*\) satisfying~\(\sect{s}{\mot{v}}=s=\sect{s}{\mot{w}}\)
and~\(\act{s}{\mot{v}}\neq\act{s}{\mot{w}}\).
We obtain~\(\alpha_t(\ell) \geq \card{\{\act{t}{\mot{x}} \mid \mot{x} \in \mot{u}(\mot{\mot{v}+\mot{w}})^* \cap \Sigma^{\ell}}\}\) for~\(\ell\geq0\). Therefore \(\alpha_t\) grows exponentially.

$(\Rightarrow)$ Assume that \(t\) has exponential activity and cannot reach two coreachable cycles in~\({\outs{\auta}_t}\). By Lemma~\ref{lem:struct}, there exist a subset~\(\mathcal{E} \subset Q\) and three words~\(\mot{u,v,w} \in \Sigma^*\) such that \(\mathcal{E}\) is the set of nontrivial states accessible from~\(t\) reading~\(\mot{u}\)  and the paths labeled by \(\mot{v}\) and \(\mot{w}\) are cycles starting from~\(\mathcal{E}\) in~\({\outs{\auta}_t}\). It means that  \(\mot{v}\) and \(\mot{w}\) are also cycles starting from \(t\) in~\({{\auta_t}}\), contradiction.
\end{proof}

Using the subadditivity of the activity (see~Lemma~\ref{lem-subadditive}), we get for the polynomial activities:

\begin{corollary}
Let \(\auta\) be a Mealy automaton. The transformations of~\(\presp{\auta}\) are all of polynomial activity if and only if there are no coreachable cycles in the automaton \(\detee{\outs{\auta}}\). Moreover the degree of the (polynomial) activity corresponds to the longest chain of cycles in \(\detee{\outs{\auta}}\) minus~1.
\end{corollary}

\begin{example}\label{e:spol} Consider the transformation~\(t_0=(\unit,t_0)[2,2]\) from~Example~\ref{e:spol-zero}, its square~\(t_0^2\), and the associated automata~\(\auta_{t_0^2}\) and~\(\detee{\outs{\auta}_{t_0^2}}\):
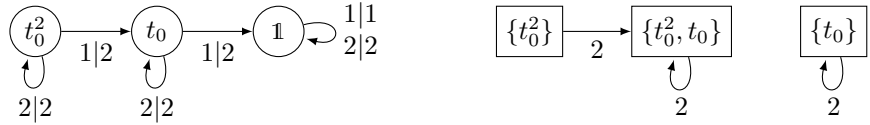
\begin{center}
\begin{tikzpicture}[->,>=latex,node distance=16mm]
	\begin{scope}[]
	\node[draw,circle,minimum width=19pt,inner sep=0pt] (tt) {$t_0^2$};
	\node[draw,circle,minimum width=19pt,inner sep=0pt,right of=tt] (t) {$t_0$};
	\node[draw,circle,minimum width=19pt,inner sep=0pt,right of=t] (one) {$\unit$};
	\path	
		(tt)	edge[loop below] 		node[below]{\(2|2\)}	(tt)
		(tt)	edge				node[below]{\(1|2\)}	(t)
		(t)	edge[loop below] 		node[below]{\(2|2\)}	(t)
		(t)	edge				node[below]{\(1|2\)}	(one)
		(one)	edge[loop right]		node[right,align=center]{\(1|1\)\\\(2|2\)}	(one);
	\end{scope}
	\begin{scope}[xshift=6.5cm,node distance=20mm]
	\node[draw,rectangle,minimum height=19pt] (tt) {$\lbrace t_0^2 \rbrace $};
	\node[draw,rectangle,minimum height=19pt,right of=tt] (ttt) {$\lbrace t_0^2,t_0 \rbrace $};
	\node[draw,rectangle,minimum height=19pt,right of=ttt] (t) {$\lbrace t_0\rbrace$};
	\path	
		(tt)	edge[] 		node[below]{\(2\)}	(ttt)

		(ttt)	edge[loop below] 		node[below]{\(2\)}	(ttt)
		(t)	edge[loop below] 		node[below]{\(2\)}	(t);
	\end{scope}
\end{tikzpicture}
\end{center}
Note that, before determinization, two disjoint cycles are accessible from the state~\(t_0^2\). In the determinized version, \(\{t_0\}\) and \(\{t_0^2\}\) both access to only one cycle, and we conclude \(\{t_0,t_0^2\} \subset \SP{0}\). By Proposition~\ref{p:closed}, we actually knew the full inclusion~\(\presp{\,t_0\,}\subset\SP{0}\).

\smallbreak
Defining further \(t_k=(t_{k-1},t_k)[2-(k \mod 2), 2-(k \mod 2)]\in\FEnd(\{1,2\}^*)\), we obtain the family with~\(t_k\in\SP{k}\smallsetminus\SP{k-1}\) for~\(k>0\),
that witnesses the strictness of the polynomial part of the hierarchy from~Theorem~\ref{t:hierarchy}.
\end{example}

Using Proposition~\ref{prop-entraut}, we obtain an explicit formula for the norm~\(\snorm{\cdot}\):

\begin{proposition}
  Let \(t\) be a finite-state transformation with associated Mealy
  automaton~\(\auta_t\).  The norm~\(\snorm{t}\) is the logarithm of
  the Perron eigenvalue of the transition matrix of
  \(\detee{\outs{\auta}_t}\):
  \[\snorm{t} = \log \lambda(\detee{\outs{\auta}_t}) \:.\]
\end{proposition}

\begin{proof} By Lemma~\ref{lem:struct}, the activity of~\(t\) counts the number of paths in~\(\detee{\outs{\auta}_t}\).
Since all its states are final, this automaton is coaccessible and the cardinality of the language accepted when putting~\(t\) as the initial state is exactly the activity of~\(t\).
Therefore by Proposition~\ref{prop-entraut}, we have
\[\snorm{t} = \lim_{\ell \to \infty} {\log \alpha_t(\ell) \over \ell} = \lim_{\ell \to \infty} {\frac{1}{\ell}{\ \log  \sum\limits_{t'=1}^{n} (A^\ell)_{t,t'}}}  = h(L) = \log  \lambda(\detee{\outs{\auta}_t}),\]
with~\(A=(A_{i,j})_{i,j}\) the adjacency matrix of~\(\detee{\outs{\auta}_t}\).
\end{proof}

\smallskip\emph{Proof of~Theorem~\ref{t:hierarchy}.}\label{proof:hierarchy} The strictness for the polynomial part is obtained from Example~\ref{e:spol}. Now, as the norm~\(\snorm{.}\) is the logarithm of the maximal eigenvalue of a matrix with integer coefficients, the classes~\(\SE{\lambda}\) increase only when~\(e^\lambda\) is an algebraic integer that is the largest zero of its minimal polynomial, \ie, a root of a Perron number. Furthermore, each of these numbers is the norm of some finite-state transformation, see~\cite[Theorem~3]{Lin84} for a proof. It is also known that Perron numbers are dense in~\([1,\infty)\), which gives us the strictness for the exponential part: $\lambda_1<\lambda_2$ implies~$\SE{\lambda_1}\subsetneq\SE{\lambda_2}$. 

Finally, the growth rate  can be computed with any precision~\(0<\delta <1\) in time~\(\Theta\left(-\log ( \delta n) \right)\), where~\(n\) is the number of states of the automaton~\cite{Shu08b}.\hfill\qed

\begin{example} Consider the transformations~$\sss=(\ttt,\sss)[1,1]$ and~$\ttt=(\unit,\sss)$ with common associated automata~\(\auta\)
(on the left) and~\(\detee{\outs{\auta}}\) (on the right):
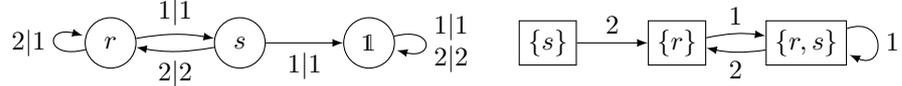
\begin{center}
\begin{tikzpicture}[->,>=latex,node distance=17mm]
	\begin{scope}
	\node[draw,circle,minimum width=19pt] (s) {$\sss$};
	\node[draw,circle,minimum width=19pt,right of=s] (t) {$\ttt$};
	\node[draw,circle,minimum width=19pt,right of=t] (one) {$\unit$};
	\path	(s)	edge[loop left] 		node[left]{\(2|1\)}	(s)
		(s)	edge[bend left=10] 	node[above]{\(1|1\)}	(t)
		(t)	edge[bend left=10] 	node[below]{\(2|2\)}	(s)
		(t)	edge				node[below]{\(1|1\)}	(one)
		(one)	edge[loop right]		node[right,align=center]{\(1|1\)\\\(2|2\)}	(one);
	\end{scope}
	\begin{scope}[xshift=57.5mm]
	\node[draw,rectangle,minimum height=16pt] (t) {$\lbrace \ttt\rbrace$};
	\node[draw,rectangle,minimum height=16pt,right of=t] (s) {$\lbrace \sss \rbrace $};
	\node[draw,rectangle,minimum height=16pt,right of=s] (st) {$\lbrace \sss,\ttt \rbrace $};
	\path	(st)	edge[loop right,xscale=.6] 	node[right]{\(1\)}	(s)
		(t)	edge[] 			node[above]{\(2\)}	(s)
		(s)	edge[bend left=10] 	node[above]{\(1\)}	(st)
		(st)	edge[bend left=10] 	node[below]{\(2\)}	(s);
	\end{scope}
\end{tikzpicture}
\end{center}
We find that $\alpha_\sss(n)$ and $\alpha_\ttt(n+1)$
correspond to the $n$-th Fibonacci number.
We deduce $\snorm{\sss}=\snorm{\ttt}=\log\varphi$ where~$\varphi$ is the golden ratio,
hence~$\sss,\ttt\in\SE{\log\varphi}$. 
\end{example}

%-------------------------------------------------------------------------------------------------------------------------
\section{The orbit signalizer graph and the order problem}\label{s:order}

This section is devoted to the \OP: can one decide whether a given
element generates a finite semigroup?  The latter is known to be
undecidable for general automaton semigroups~\cite{Gil14} and
decidable for~\(\Pol(0)\)~\cite{BBSZ13}. We give a general
construction that associates a graph to a transformation of
$\Sigma^*$, and show that, if finite, this graph lets us compute the
index and period of the transformation.  We show that this graph is
finite for elements from~\(\SPol(0)\), and solve the \OP in this
manner.

\medbreak Let $\Sigma$ be an alphabet. We define the \emph{orbit signalizer graph}~$\Phi$ for~$\FEnd(\Sigma^*)$ as the following (infinite) graph. The vertices are the pairs of elements in~$\FEnd(\Sigma^*)$.
For each letter~$x\in\Sigma$, there is an arrow from the source~$(s,t)$ with label~$(\lab{x}{\dindex\,}{\dperiod})$
where~$\dindex$ and~$\dperiod$ are the minimal integers (with~\(\dperiod>0\)) satisfying\[\act{st^{\dindex+\dperiod}}{x}=\act{st^{\dindex}}{x},\]
and with target~\((\sect{r}{x},\sect{t^\dperiod}{\act{r}{x}})\) for~\(r=st^{\dindex}\).
The parameters~$\dindex$ and~$\dperiod$ correspond respectively to the \emph{index}
and to the~\emph{period} of the orbit of~$x$ under the action of~$st^\omega$, see~Figure~\ref{fig-frypan}. % and Section~\ref{sec-ex}.

In what follows, the intuition is roughly to generalize Figure~\ref{fig-frypan}, by considering a path~\(\pi\) instead of the letter~\(x\): such a construction leads also to a pan graph, whose handle has length between~\(i^-_{t}\) and~\(i^+_{t}\), and whose cycle has length~\(p_t\). The main challenge here is to be able to keep the construction finite, when possible.

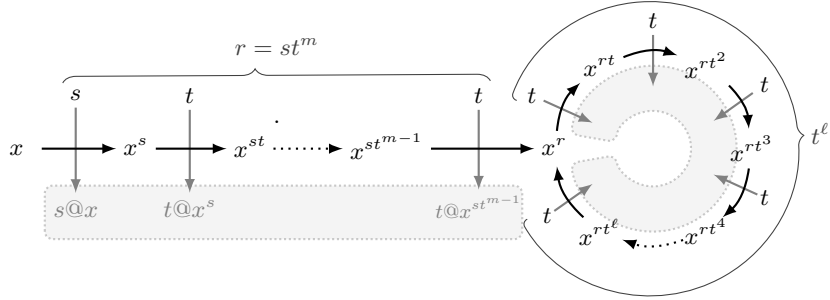
\begin{figure}[ht]
\vspace*{-12pt}
\begin{center}
\begin{tikzpicture}[->,>=latex,node distance=12mm]
\begin{scope}[inner sep= 2.5pt]
	\node (Y0) {$x^r_{\phantom{5}}$};
	\node (Z) [node distance=13mm,right of=Y0] {};
	\begin{scope}[line width=.8pt,rounded corners=3pt,densely dotted]
		\draw[fill=gray!40,opacity=.2] (-6.7,-1.2) rectangle (-.42,-.49);
		\draw[fill=gray!40,opacity=.2] (Z) ++(-170:1.1)
     			arc(-170:170:1.1)      -- ++(170:-0.6)
     			arc(170:-170:0.5)    -- cycle;
	\end{scope}
	\node (t0) [node distance=11mm,above of=Y0] {};
	\node (Y1) [node distance=6mm,right of=t0] {$\act{rt}{x}$};
	\node (Y2) [node distance=14mm,right of=Y1]{$\act{rt^2}{x}$};
	\node (t1) [node distance=11mm,below of=Y2] {};
	\node (Y3) [node distance=6mm,right of=t1] {$\act{rt^3}{x}$};
	\node (Y4) [node distance=11mm,below of=t1,inner sep=.4pt] {$\act{rt^4}{x}$};
	\node (Y5) [node distance=14mm,left of=Y4,inner sep=.3pt] {$\act{rt^{\dperiod}}{x}$};
	\node (X4) [node distance=22mm,left of=Y0] {$\act{st^{\dindex-1}}{x}_{\phantom{5}}$};
	\node (X3) [node distance=18mm,left of=X4] {$\act{st}{x}_{\phantom{5}}$};
	\node (X2) [node distance=15mm,left of=X3] {$\act{s}{x}_{\phantom{5}}$};
	\node (X1) [node distance=15mm,left of=X2] {$\act{\phantom{r}}{x}_{\phantom{5}}$};
	\node (T1) [node distance=10mm,above right of=X1] {$s$};
	\node (A1) [node distance=15mm,below of=T1,gray] {$\sect{s}{x}$};
	\node (T2) [node distance=10mm,above right of=X2] {$t$};
	\node (A2) [node distance=15mm,below of=T2,gray] {$\sect{t}{\act{s}{x}}$};
	\node (T4) [node distance=38mm,right of=T2] {$t$};
	\node (A4) [node distance=38mm,right of=A2,gray,scale=.85] {$\sect{t}{\act{{st^{\dindex-1}}}{x}_{\phantom{5}}}$};
	\node (T6) [node distance=17mm,above of=Z] {$t$};
	\node (A6) [node distance=9mm,below of=T6,inner sep=.4pt] {};
	\node (T5) [node distance=7mm,right of=T4] {$t$};
	\node (A5) [node distance=8mm,below of=Y1] {};
	\node (T7aux) [node distance=32mm,right of=T5] {};
	\node (T7) [node distance=2mm,above left of=T7aux] {$t$};
	\node (A7) [node distance=14mm,right of=A5] {$$};
	\node (T8) [node distance=15mm,below of=T7,inner sep=.8pt] {$t$};
	\node (A8) [node distance=6mm,below of=A7] {$$};
	\node (T9aux) [node distance=16mm,below of=T5] {};
	\node (T9) [node distance=2mm,right of=T9aux,inner sep=.8pt] {$t$};
	\node (A9) [node distance=7mm,below of=A5] {$ $};
	\node (T3) [node distance=5mm,above right of=X3] {.};
	\node (R) [node distance=10mm,above of=T3,darkgray] {$r=st^{\dindex}$};
	\path[thick]
		(T1)	edge	[gray]					(A1)
		(T2)	edge	[gray]					(A2)
		(T4)	edge	[gray]					(A4)
		(T5)	edge	[gray]					(A5)
		(T6)	edge	[gray]					(A6)
		(T7)	edge	[gray]					(A7)
		(T8)	edge	[gray]					(A8)
		(T9)	edge	[gray]					(A9)
		(X1)	edge							(X2)
		(X2)	edge							(X3)
		(X3)	edge	[dotted]					(X4)
		(X4)	edge							(Y0)
		(Y0)	edge[bend left=20]				(Y1)
		(Y1)	edge[bend left=20]				(Y2)
		(Y2)	edge[bend left=20]				(Y3)
		(Y3)	edge[bend left=20]				(Y4)
		(Y4)	edge[bend left=20, dotted]			(Y5)
		(Y5)	edge[bend left=20]				(Y0);
	\begin{scope}[-]
		\draw[darkgray,decorate,decoration={brace,amplitude=5pt,aspect=.45,angle=20}] (T1.north west) -- (T4.north east);
		\curlybrace[color=darkgray,xshift=18pt]{-150}{160}{1.95} \node[darkgray,anchor=west] at (curlybracetipn) {$t^{\dperiod}_{\phantom{5}}$};
	\end{scope}
\end{scope}
\end{tikzpicture}
\end{center}
\vspace*{-12pt}
\caption{The cross-diagram associated with the orbit of some letter~$x\in\Sigma$ under the action of~$st^\omega$. The index~\(\dindex\) and period~\(\dperiod\) will complete the label of the \(x\)-arrow away from the vertex~\((s,t)\) in the graph~\(\Phi\). Each of the two gray zones indicates an entry of the corresponding target vertex~\((\sect{r}{x},\sect{t^\dperiod}{\act{r}{x}})\) with~\(r=st^{\dindex}\).}
\label{fig-frypan}
\end{figure}

The \emph{inf-index-cost}, \emph{sup-index-cost}, and the \emph{period-cost} of a given walk~\(\pi\) on~$\Phi$
\[\pi:(s,t)\xrightarrow{\lab{x_1}{\,\dindex_1}{\dperiod_1}}\cdots\xrightarrow{\lab{x_{|\pi|}}{\,\dindex_{|\pi|}}{\dperiod_{|\pi|}}}(s',t')\]
are respectively defined by
\[i^-(\pi)=\sum_{1\leq k\leq |\pi|}\!\!\left((1-\delta_{m_k,0})\left( (\dindex_{k}-1) \left(\prod_{1\leq j<k} \!\! \dperiod_j\right)+1\right)\right),\]
\[i^+(\pi)=\sum_{1\leq k\leq |\pi|}\!\! \left(\prod_{1\leq j<k} \!\! \dperiod_j\right)\dindex_{k},
\hbox{\qquad and\qquad} p(\pi)=\!\! \prod_{1\leq i\leq |\pi|}\!\!  \dperiod_i\ .\]

\medbreak For any~\(t\in\FEnd(\Sigma^*)\), we define the \emph{orbit signalizer graph}~\(\Phi(t)\) as the subgraph
of~\(\Phi\) accessible from the source vertex~\((\unit,t)\). 
The \emph{inf-index-cost}, \emph{sup-index-cost}, and the \emph{period-cost} of~$t\in\FEnd(\Sigma^*)$
are then respectively defined by
\[i^-_{t}=\max_{\pi\hbox{\scriptsize~on~}\Phi(t)}i^-(\pi),
\qquad
i^+_{t}=\max_{\pi\hbox{\scriptsize~on~}\Phi(t)}i^+(\pi),
\hbox{\qquad and\qquad}
p_t=\underset{\pi\hbox{\scriptsize~on~}\Phi(t)}{\lcm}p(\pi)\ .\]

\begin{proposition}\label{prop-eitheror} The semigroup generated by an element~\(t\in\FEnd(\Sigma^*)\)
is finite if and only if its index-costs~\(i^\pm_{t}\) and its period-cost~\(p_t\) are finite.
In that case, we have~\(\presp{\,t\,}=\presp{\,t:t^{i_t}=t^{i_t+p_t}\,}\) for some index~\(i^{\phantom{\pm}\!\!\!}_{t}\)
with~\(i^{-}_{t}\leq i^{\phantom{\pm}\!\!\!}_{t}\leq i^{+}_{t}\).
\end{proposition}

\begin{proof} Let~\(\Sigma=\{x_1,\ldots,x_{|\Sigma|}\}\). Let~\((s_0,t_0)\) be a vertex in~\(\Phi\) and~\((s_k,t_k)\) its successor vertex with arrow~\(\lab{x_k}{m_k}{\ell_k}\) for~\(1\leq k\leq|\Sigma|\).
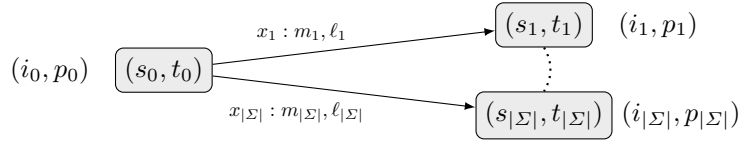
\begin{center}
\begin{tikzpicture}[->,>=latex]
\begin{scope}[rounded corners=3pt,node distance=6mm,minimum width=12mm]
	\node[draw,rectangle,fill=gray!15] (X0) {\((s_0,t_0)\)};
	\node[node distance=50mm,right of=X0] (Xk) {};
	\node[draw,rectangle,fill=gray!15,above of=Xk] (X1) {\((s_1,t_1)\)};
	\node[draw,rectangle,fill=gray!15,below of=Xk] (XX) {\((s_{|\Sigma|},t_{|\Sigma|})\)};
	\node[node distance=15mm,left of=X0] (X0ip) {\((i_0,p_0)\)};
	\node[node distance=15mm,right of=X1] (X0ip) {\((i_1,p_1)\)};
	\node[node distance=18mm,right of=XX] (X0ip) {\((i_{|\Sigma|},p_{|\Sigma|})\)};
	\path	(X0)	edge		node[above left,	align=center,scale=.75]{\(\lab{x_1}{m_1}{\ell_1}\)}					(X1);
	\path	(X0)	edge		node[below left,	align=center,scale=.75,pos=.6]{\(\lab{x_{|\Sigma|}}{m_{|\Sigma|}}{\ell_{|\Sigma|}}\)}	(XX);
	\path	(X1.south)	 edge[-,dotted,thick,bend left]	(XX.north);
\end{scope}
\end{tikzpicture}
\end{center}
For~\(0\leq k\leq|\Sigma|\),
let~\((i_k,p_k)\in\{\omega,0,1,2,\ldots\}\times\{\omega,1,2,3\ldots\}\)
denote the possible minimal pair of ordinals
(with~\(p_k>0\)) satisfying\[s_kt_k^{i_k}=s_kt_k^{i_k+p_k}.\] Whenever
there is at least one successor with~\((i_k,p_k)=(\omega,\omega)\),
\((s_0,t_0)\) satisfies also~\((i_0,p_0)=(\omega,\omega)\), and so
does any of its predecessors.  Otherwise, we claim
\[\max_{1\leq k\leq |\Sigma|}\big(m_k+\max(0,\ell_k(i_k -1)+1)\big)\leq i_0\leq \max_{1\leq k\leq |\Sigma|}\left(m_k+\ell_k i_k\right)\]
and\[p_0=\underset{1\leq k\leq |\Sigma|}{\lcm}\ell_k p_k.\]
Indeed,  for~\(1\leq k\leq |\Sigma|\) and for any~\(u\in\Sigma^*\),
we have\[y_kv=\act{s_0t_0^{\dindex_k+\dperiod_ki_k}}{(x_ku)}=\act{t_0^{\dperiod_kp_k}}{(y_kv)}\]
with~\(y_k=\act{s_0t_0^{\dindex_k}}{x_k}\) and~\(v=\act{s_kt_k^{i_k+p_k}}{u}\),
as illustrated by the
cross-diagram:
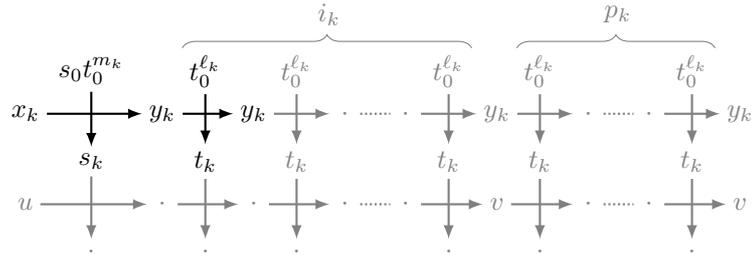
\begin{center}
\begin{tikzpicture}[->,>=latex]
\begin{scope}[inner sep= 2.5pt]
	\node[gray] (X5) {$y_{k}$};
	\node (X4) [node distance=12mm,left of=X5,gray] {$\,\cdot\,$};
	\node (X3) [node distance=8mm,left of=X4,gray] {$\,\cdot\,$};
	\node (X2) [node distance=12mm,left of=X3] {$y_{k}$};
	\node (X1) [node distance=12mm,left of=X2] {$y_{k}$};
	\node (X6) [node distance=12mm,right of=X5,gray] {$\,\cdot\,$};
	\node (X7) [node distance=8mm,right of=X6,gray] {$\,\cdot\,$};
	\node (X8) [node distance=12mm,right of=X7,gray] {$y_{k}$};
	\node (X0) [node distance=18mm,left of=X1] {$x_{k}$};
	\node (T1) [node distance=8mm,above right of=X1] {$t_0^{\dperiod_k}$};
	\node (A1) [node distance=12mm,below of=T1] {$t_k$};
	\node (T2) [node distance=8mm,above right of=X2,gray] {$t_0^{\dperiod_k}$};
	\node (A2) [node distance=12mm,below of=T2,gray] {$t_k$};
	\node (T3) [node distance=8mm,above right of=X3] {};
	\node (A3) [node distance=12mm,below of=T3] {};
	\node (T4) [node distance=8mm,above right of=X4,gray] {$t_0^{\dperiod_k}$};
	\node (A4) [node distance=12mm,below of=T4,gray] {$t_k$};
	\node (T5) [node distance=8mm,above right of=X5,gray] {$t_0^{\dperiod_k}$};
	\node (A5) [node distance=12mm,below of=T5,gray] {$t_k$};
	\node (T6) [node distance=8mm,above right of=X6] {};
	\node (A6) [node distance=12mm,below of=T6,gray] {};
	\node (T7) [node distance=8mm,above right of=X7,gray] {$t_0^{\dperiod_k}$};
	\node (A7) [node distance=12mm,below of=T7,gray] {$t_k$};
	\node (T0) [node distance=8mm,above right of=X0,xshift=3mm] {$s_0t_0^{\dindex_k}$};
	\node (A0) [node distance=12mm,below of=T0] {$s_k$};
	\node (XX5) [node distance=12mm,below of=X5,gray]{$v$};
	\node (XX4) [node distance=12mm,left of=XX5,gray] {$\,\cdot\,$};
	\node (XX3) [node distance=8mm,left of=XX4,gray] {$\,\cdot\,$};
	\node (XX2) [node distance=12mm,left of=XX3,gray] {$\,\cdot\,$};
	\node (XX1) [node distance=12mm,left of=XX2,gray] {$\,\cdot\,$};
	\node (XX6) [node distance=12mm,right of=XX5,gray] {$\,\cdot\,$};
	\node (XX7) [node distance=8mm,right of=XX6,gray] {$\,\cdot\,$};
	\node (XX8) [node distance=12mm,right of=XX7,gray] {$v$};
	\node (XX0) [node distance=18mm,left of=XX1,gray] {$u$};
	\node (AA1) [node distance=12mm,below of=A1,gray] {$\,\cdot\,$};
	\node (AA2) [node distance=12mm,below of=A2,gray] {$\,\cdot\,$};
	\node (AA3) [node distance=12mm,below of=A3] {$$};
	\node (AA4) [node distance=12mm,below of=A4,gray] {$\,\cdot\,$};
	\node (AA5) [node distance=12mm,below of=A5,gray] {$\,\cdot\,$};
	\node (AA6) [node distance=8mm,below of=A6] {$$};
	\node (AA7) [node distance=12mm,below of=A7,gray] {$\,\cdot\,$};
	\node (AA0) [node distance=12mm,below of=A0,gray] {$\,\cdot\,$};
	\node (R) [node distance=10mm,above of=T3,gray,yshift=-7pt,xshift=-22pt] {$i_k$};
	\node (R2) [node distance=10mm,above of=T6,gray,yshift=-7pt,xshift=-5pt] {$p_k$};
	\path[thick]
		(T0)	edge							(A0)
		(T1)	edge							(A1)
		(T2)	edge	[gray]					(A2)
		(T4)	edge	[gray]					(A4)
		(T5)	edge	[gray]					(A5)
		(T7)	edge	[gray]					(A7)
		(X0)	edge							(X1)
		(X1)	edge							(X2)
		(X2)	edge	[gray]					(X3)
		(X3)	edge	[-,densely dotted,gray]		(X4)
		(X4)	edge	[gray]					(X5)
		(X5)	edge	[gray]					(X6)
		(X6)	edge	[-,densely dotted,gray]		(X7)
		(X7)	edge	[gray]					(X8)
		(A0)	edge	[gray]					(AA0)
		(A1)	edge	[gray]					(AA1)
		(A2)	edge	[gray]					(AA2)
		(A4)	edge	[gray]					(AA4)
		(A5)	edge	[gray]					(AA5)
		(A7)	edge	[gray]					(AA7)
		(XX0)	edge	[gray]					(XX1)
		(XX1)	edge	[gray]					(XX2)
		(XX2)	edge	[gray]					(XX3)
		(XX3)	edge	[-,densely dotted,gray]		(XX4)
		(XX4)	edge	[gray]					(XX5)
		(XX5)	edge	[gray]					(XX6)
		(XX6)	edge	[-,densely dotted,gray]		(XX7)
		(XX7)	edge	[gray]					(XX8)
;
	\begin{scope}[-]
		\draw[gray,decorate,decoration={brace,amplitude=5pt,aspect=.5,angle=20}] (T1.north west) -- (T4.north east);
		\draw[gray,decorate,decoration={brace,amplitude=5pt,aspect=.5,angle=20}] (T5.north west) -- (T7.north east);
	\end{scope}
\end{scope}
\end{tikzpicture}
\end{center}
We conclude using an induction on the length of the paths.
\end{proof}

%%-------------------------------------------------------------------------------------------------------------------------

\begin{theorem}
The \OP\ is decidable for any~\(t \in \FE\) with a finite orbit signalizer graph~\(\Phi(t)\).
\end{theorem}
\begin{proof}
Since \(\Phi(t)\) is a graph with outdegree~\(\size{\Sigma}>0\) by construction, its finiteness implies the existence of cycles.
Consider the simple cycles (there is only a finite number of these).
One can compute the index-costs~\(i^-(\kappa) \) and~\(i^+(\kappa)\) and the period-cost~\(p(\kappa)\) of each such cycle~\(\kappa\).
Whenever \(i^-(\kappa)>0 \) or \(p(\kappa)>1\) for some cycle~\(\kappa\), then \(t\) has infinite order, and finite order otherwise. 
\end{proof}

\begin{example}\label{e:semST}  The transformations~$s=(s,\unit)[2,2]$ and~$t_0=(\unit,t_0)[2,2]$ (on the left)
admit respective graphs~$\Phi(s)$ and~$\Phi(t_0)$ (on the right):
\begin{center}
\begin{tikzpicture}[->,>=latex,node distance=6mm,minimum size=6mm]
	\begin{scope}
	\node[draw,circle,minimum width=19pt,inner sep=0pt] (s) {$s$};
	\node[below of=s] (aux) {};
	\node[draw,circle,minimum width=19pt,below of=aux,inner sep=0pt] (t) {$t_0$};
	\node[draw,circle,minimum width=19pt,node distance=20mm,right of=aux,inner sep=0pt] (one) {$\unit$};
	\path	(s)	edge[loop left] 		node[left]{\(1|2\)}	(s)
		(s)	edge				node[above]{\(2|2\)}	(one)
		(t)	edge[loop left] 		node[left]{\(2|2\)}	(t)
		(t)	edge				node[below]{\(1|2\)}	(one)
		(one)	edge[loop right]		node[right,align=center]{\(1|1\)\\\(2|2\)}	(one);
	\end{scope}
	\begin{scope}[xshift=6cm,rounded corners=3pt,node distance=23mm]
		\node[draw,rectangle,fill=gray!15] (1s) {$(\unit,s)$};
		\node[draw,rectangle,right of=1s,fill=gray!15] (11) {$(\unit,\unit)$};
		\node[draw,rectangle,node distance=12mm,below of=1s,fill=gray!15] (s1) {$(s,\unit)$};
		\node[draw,rectangle,right of=s1,fill=gray!15] (1t) {$(\unit,t_0)$};
		\path	(1s)	edge 			node[above,align=center,scale=.75]{\(\labb{2}{0}{1}\)}			(11)
			(1s)	edge 			node[left,align=center,scale=.75]{\(\labb{1}{1}{1}\)}			(s1)
			(s1)	edge[bend left=15] 	node[right=.1,pos=.15,align=center,scale=.75]{\(\labb{2}{0}{1}\)}	(11)
			(s1)	edge[loop left] 		node[left,align=center,scale=.75]{\(\labb{1}{0}{1}\)}			(s1)
			(11)	edge[loop right] 	node[right,align=center,scale=.75]{\(\labb{1}{0}{1}\)\\\(\labb{2}{0}{1}\)}	(11)
			(1t)	edge[loop right] 	node[right,align=center,scale=.75]{\(\labb{1}{1}{1}\)\\\(\labb{2}{0}{1}\)}	(1t);
	\end{scope}
\end{tikzpicture}
\end{center}
According to Proposition~\ref{prop-eitheror}, they generate the finite monoid~$\langle~s:s^2=s~\rangle_+$
and the free monoid~$\langle~t_0:~\ ~\rangle_+$.
\end{example}

\begin{example}\label{e:sem393} 
%sem393 := MealyMachine([[1,1,1],[1,1,2],[2,1,3]],[[1,2,3],[1,1,2],[2,3,1]]);
The transformation~\(b=(a,\unit,b)[2,3,1]\) from~\(\SPol(1)\smallsetminus\SPol(0)\)
with~$a=(\unit,\unit,a)[1,1,2]$ admits the finite graph~$\Phi(b)$ displayed on Figure~\ref{f:sem393},
in which we can read that both~$ab$ and~$ba$ have period~$1$, and that $b$ has thus period~$p_b=3$.
According to Proposition~\ref{prop-eitheror} again, the index of~$b$ satisfies~$4\leq i_b\leq 9$,
and can be explicitly computed as~$i_b=8$.
\end{example}

\begin{figure}
\begin{center}
\begin{tikzpicture}[->,>=latex,node distance=20mm,inner sep=0,minimum size=7mm]
	\begin{scope}
	\node[draw,circle,minimum width=19pt] (a) {$a$};
	\node[draw,circle,minimum width=19pt,below of=a] (b) {$b$};
	\node[draw,circle,minimum width=19pt,right of=a] (one) {$\unit$};
	\node[draw,circle,minimum width=19pt,left of=a,gray!85] (ab) {$ab$};
	\node[draw,circle,minimum width=19pt,left of=b,gray!85] (ba) {$ba$};
	\path	(a)	edge[loop above] 		node[above]{\(3|2\)}	(a)
		(a)	edge					node[align=center,above]{\(1|1\)\\\(2|1\)}	(one)
		(b)	edge[loop below] 		node[below]{\(3|1\)}	(b)
		(b)	edge				node[right]{\(1|2\)}	(a)
		(b)	edge				node[below right,pos=.4]{\(2|3\)}	(one)
		(one)	edge[loop right]		node[right,align=center]{\(1|1\)\\\(2|2\)\\\(3|3\)}	(one)
		(ab)	edge[gray!85]		node[above,align=center]{\(1|2\)\\\(2|2\)\\\(3|3\)}	(a)
		(ba)	edge[gray!85]		node[above left,pos=.35,align=center]{\(1|1\)\\\(2|2\)}	(a)
		(ba)	edge[gray!85]		node[below,align=center]{\(3|1\)}	(b);
	\end{scope}
	\begin{scope}[xshift=60mm,yshift=10mm,rounded corners=3pt,node distance=21mm,minimum size=6mm,minimum width=11mm]
		\node[draw,rectangle,fill=gray!15] (1b) {$(\unit,b)$};
		\node[draw,rectangle,fill=gray!15,below left of=1b] (1ab) {$(\unit,ab)$};
		\node[draw,rectangle,fill=gray!15,below right of=1b] (1ba) {$(\unit,ba)$};
		\node[draw,rectangle,fill=gray!15,below left of=1ab] (aa) {$(a,a)$};
		\node[draw,rectangle,fill=gray!15,below left of=1ba] (1a) {$(\unit,a)$};
		\node[draw,rectangle,fill=gray!15,below right of=1ba] (ba) {$(b,a)$};
		\begin{scope}[node distance=15mm]
		\node[draw,rectangle,fill=gray!15,below of=1a] (a1) {$(a,\unit)$};
		\node[draw,rectangle,fill=gray!15,below of=ba] (b1) {$(b,\unit)$};
		\node[draw,rectangle,fill=gray!15,below of=aa] (11) {$(\unit,\unit)$};
		\end{scope}
		\path	(1b)	edge		node[above left,		align=center,scale=.75]{\(\labb{1}{0}{3}\)}				(1ab)
			(1b)	edge 	node[above right,		align=center,scale=.75]{\(\labb{2}{0}{3}\)\\\(\labb{3}{0}{3}\)}	(1ba)
			(1ab)	edge		node[above right,		align=center,pos=.4,scale=.75]{\(\labb{2}{0}{1}\)\\\(\labb{3}{0}{1}\)}	(1a)
			(1ab)	edge		node[above left,		align=center,scale=.75]{\(\labb{1}{1}{1}\)}	(aa)
			(1ba)	edge[] 	node[above left,		align=center,pos=.4,scale=.75]{\(\labb{1}{0}{1}\)\\\(\labb{2}{0}{1}\)}	(1a)
			(1ba)	edge[] 	node[above right,		align=center,scale=.75]{\(\labb{3}{1}{1}\)}			(ba)
			(1a)	edge[] 	node[right=0.01,		align=center,scale=.75]{\(\labb{3}{2}{1}\)}			(a1)
			(1a)	edge[] 	node[above=.2,			align=center,pos=.28,scale=.75]{\(\labb{1}{0}{1}\)\\\(\labb{2}{1}{1}\)}	(11)
			(aa)	edge[] 	node[above=.1,			align=center,pos=.3,scale=.75]{\(\labb{3}{0}{1}\)}	(a1)
			(aa)	edge[] 	node[left=.1,			align=center,scale=.75]{\(\labb{1}{0}{1}\)\\\(\labb{2}{0}{1}\)}	(11)
			(ba)	edge[] 	node[above left,		align=center,pos=.3,scale=.75]{\(\labb{1}{1}{1}\)\\\(\labb{2}{2}{1}\)}	(a1)
			(ba)	edge[] 	node[right,			align=center,scale=.75]{\(\labb{3}{0}{1}\)}	(b1)
			(a1)	edge[loop below] 	node[below,	align=center,scale=.75]{\(\labb{3}{0}{1}\)}	(a1)
			(a1)	edge		node[below=.15,		align=center,scale=.75]{\(\labb{1}{0}{1}\)\\\(\labb{2}{0}{1}\)}	(11)
			(b1)	edge[loop below]	node[below,	align=center,scale=.75]{\(\labb{3}{0}{1}\)}	(b1)
			(b1)	edge		node[below,			align=center,scale=.75]{\(\labb{1}{0}{1}\)}	(a1)
			(b1)	edge[bend left=35]	node[below=.01,	align=center,scale=.75]{\(\labb{2}{0}{1}\)}	(11)
			(11)	edge[loop below]	node[below=.1,	align=center,scale=.75]{\(\labb{1}{0}{1}\)\\\(\labb{2}{0}{1}\)\\\(\labb{3}{0}{1}\)}	(11);
	\end{scope}
\end{tikzpicture}
\end{center}
\caption{The Mealy automaton~\(\auta_b\) and the graph~\(\Phi(b)\) from Example~\ref{e:sem393}.}
\label{f:sem393}
\end{figure}

\begin{proposition} Every bounded finite-state
transformation~$t\in \SP{0}$ admits a finite orbit signalizer graph~$\Phi(t)$.
\end{proposition}

\begin{proof}
The activity~$\alpha_t$ of~\(t\in \SP{0}\) is uniformly bounded by some constant~\(C\):
\[\#\{\mot{v}\in\Sigma^n:\exists \mot{u}\in\Sigma^n,~\sect{t}{\mot{u}}\not=\unit~\hbox{and}~\act{t}{\mot{u}}=\mot{v}\} \leq C\hbox{\quad for\quad}n\geq0\ .\]
Now the vertices of the graph~\(\Phi(t)\) are those pairs~\(\left( r,s \right)\) with~\(rs=\sect{t^K}{\mot{u}}\),
where~\(K\) is the greatest integer such that the images \(\act{{t^i}}{\mot{u}}\) for~\(i< K\) are pairwise different.
Hence \(rs\) is the product of at most \(C\) nontrivial elements. Moreover these elements \(\sect{t^i}{\mot{u}}\)
for~\(i <  K\) lie in a finite set, as \(t\) is a finite-state transformation, hence the vertices belong to a finite set. We conclude that the graph~\(\Phi(t)\) is finite.
\end{proof}

\begin{corollary} The  \OP\ is decidable for
\(\SPol(0)\).
\end{corollary}

\bibliographystyle{amsplain}

\bibliography{bibliSPol}

\end{document}